\documentclass[runningheads,10pt]{llncs}

\usepackage{graphicx}
\usepackage{amssymb}
\usepackage{amsmath}
\usepackage{comment}
\usepackage[disable]{todonotes}
\usepackage[utf8]{inputenc}
\usepackage{shuffle}
\usepackage{multirow}
\usepackage{listings}
\usepackage{mathabx}
\usepackage{wrapfig}

\usepackage{thm-restate}

\usepackage{hyperref}

\usepackage{cite}

\usetikzlibrary{arrows,shapes,automata,positioning}

\newcommand{\pref}{\operatorname{Pref}}
\newcommand{\suff}{\operatorname{Suff}}
\newcommand{\factor}{\operatorname{Fact}}

\newcommand{\NP}{\textsf{NP}}
\newcommand{\PSPACE}{\textsf{PSPACE}}

\newcommand{\PTIME}{\textsf{P}}

\newtheorem{fact}{Fact}
\newcommand{\bin}{\ensuremath{\operatorname{bin}}}
\DeclareMathOperator{\Orb}{\mathrm{Orb}}

\title{On a Class of Constrained Synchronization Problems in \NP}
\author{Stefan Hoffmann}
\authorrunning{S. Hoffmann}
\institute{Informatikwissenschaften, FB IV, 
  Universit\"at Trier,  Universitätsring 15, 54296~Trier, Germany, 
  \email{hoffmanns@informatik.uni-trier.de}}
\usepackage{tabularx,lipsum,environ,amsmath,amssymb}
\newcounter{problemcounter}
\makeatletter
\newcommand{\problemtitle}[1]{\gdef\@problemtitle{#1}}
\newcommand{\probleminput}[1]{\gdef\@probleminput{#1}}
\newcommand{\problemquestion}[1]{\gdef\@problemquestion{#1}}
\NewEnviron{decproblem}{
  \refstepcounter{problemcounter}
  \problemtitle{}\probleminput{}\problemquestion{}
  \BODY
  \par\addvspace{.5\baselineskip}
  \noindent
  \begin{tabularx}{\textwidth}{@{\hspace{\parindent}} l X c}
    \multicolumn{2}{@{\hspace{\parindent}}l}{\textbf{Decision Problem \theproblemcounter:} \@problemtitle} \\
    \textbf{Input:} & \@probleminput \\
    \textbf{Question:} & \@problemquestion
  \end{tabularx}
  \par\addvspace{.5\baselineskip}
}
\makeatother

\begin{document}

\maketitle

\begin{abstract}

  We characterize a class of constraint automata
  that gives constrained problems in $\NP$, which encompasses
  all known constrained synchronization problems in $\NP$ so far.
  We call these automata polycyclic automata. The corresponding
  language class of polycyclic languages is introduced. We show various characterizations
  and closure properties for this new language class.
  We then give a criterion for \NP-completeness and a criterion
  for polynomial time solvability for polycyclic constraint languages.
  
\end{abstract}

\section{Introduction}
\label{sec:introduction}

\todo{
noch etwas geschichte. subset sync varianten und so

Skizze NP-Reduktion.

QED an alle Beweise.

erwähnen, dass nichts mit polycyclic monoid zu tun hat.
}

A deterministic semi-automaton is synchronizing if it admits a reset word, i.e., a word which leads to some definite
state, regardless of the starting state. This notion has a wide range of applications, from software testing, circuit synthesis, communication engineering and the like, see \cite{Vol2008,San2005}.  The famous \v{C}ern\'y conjecture \cite{Cer64}
states that a minimal synchronizing word has length at most $(n-1)^2$ for an $n$-state automaton.
We refer to the mentioned survey articles~\cite{Vol2008,San2005} for details.

Due to its importance, the notion of synchronization has undergone a range of generalizations and variations
for other automata models.
It was noted in \cite{Martyugin12} that in some  generalizations only certain paths, or input words, are allowed (namely those for which the input automaton is defined). In \cite{Gusev:2012}
the notion of constrained synchronization was 
introduced in connection with a reduction procedure
for synchronizing automata.
The paper~\cite{DBLP:conf/mfcs/FernauGHHVW19} introduced the computational problem of constrained
synchronization. In this problem, we search for a synchronizing word coming from some specific subset of allowed
input sequences. For further motivation and applications we refer to the aforementioned paper~\cite{DBLP:conf/mfcs/FernauGHHVW19}.

Let us mention that restricting the solution space by a regular language
has also been applied in other areas, for example to topological sorting~\cite{DBLP:conf/icalp/AmarilliP18},
solving word equations~\cite{DBLP:journals/iandc/DiekertGH05,Diekert98TR}, constraint programming~\cite{DBLP:conf/cp/Pesant04}, or
shortest path problems~\cite{DBLP:journals/ipl/Romeuf88}.
In~\cite{DBLP:conf/mfcs/FernauGHHVW19} it was shown that the smallest partial
 constraint automaton for which the problem becomes $\PSPACE$-complete
 has two states and a ternary alphabet.
 Also, the smallest constraint automaton for which the problem
 is $\NP$-complete needs three states and a binary alphabet.
 A complete classification of the complexity landscape 
 for constraint automata with two states and a binary or ternary alphabet
 was given in~\cite{DBLP:conf/mfcs/FernauGHHVW19}. 
 In~\cite{DBLP:journals/corr/abs-2005-05907} the result for two-state automata
 was generalized to arbitrary alphabets, and a complexity classification for special 
 three-state constraint automata over a binary alphabet was given.
 As shown in~\cite{DBLP:journals/corr/abs-2005-04042}, for regular commutative constraint
 languages we only find constrained problems that are $\NP$-complete, 
 $\PSPACE$-complete, or solvable in polynomial time.
 In all the mentioned work~\cite{DBLP:conf/mfcs/FernauGHHVW19,DBLP:journals/corr/abs-2005-04042,DBLP:journals/corr/abs-2005-05907},
 it was noted that the constraint automata for which the corresponding
 constrained synchronization problem is $\NP$-complete
 admit a special form, which we generalize in this work.
 
 \medskip
 
 \noindent\textbf{Our contribution:} Here, we generalize a theorem from~\cite{DBLP:conf/mfcs/FernauGHHVW19} 
 to give a wider class of constrained synchronization problems in~$\NP$.
 As noted in~\cite{DBLP:journals/corr/abs-2005-05907}, the constraint automata that yield problems
 in~$\NP$ admit a special form and our class encompasses all known cases of constrained problems in~$\NP$.
 We also give a characterization that this class is given  precisely by those constraint automata
 whose strongly connected components are single cycles.
 We call automata of this type \emph{polycyclic}. Then we introduce the language class of polycyclic
 languages. We show that this class is closed under union, quotients, concatenation and also admits certain robustness properties with respect to different
 definitions by partial or nondeterministic automata.
 Lastly, we also give a criterion for our class that yields constrained synchronization
 problems that are $\NP$-complete and a criterion for problems in~$\PTIME$.
 
\section{Preliminaries and Definitions}
\label{sec:preliminaries}

By $\mathbb N = \{0,1,2,\ldots\}$ we denote the natural numbers, including zero.
Throughout the paper, we consider \emph{deterministic finite automata (DFAs)}.
Recall that a DFA~$\mathcal A$ is a tuple
$\mathcal A = (\Sigma, Q, \delta, q_0, F)$,
where the alphabet $\Sigma$ is a finite set of input symbols,~$Q$ is the finite state set, with start state $q_0 \in Q$, and final state set $F \subseteq Q$.
The transition function $\delta : Q\times \Sigma \to Q$ extends to words from $\Sigma^*$ in the usual way. The function $\delta$ can be further extended to sets of states in the following way. For every set $S \subseteq Q$ with $S \neq \emptyset$ and $w \in \Sigma^*$, we set $\delta(S, w) := \{\,\delta(q, w) \mid q \in S\,\}$.
We call $\mathcal A$ \emph{complete} if~$\delta$ is defined for every $(q,a)\in Q \times \Sigma$; if $\delta$ is undefined for some $(q,a)$, the automaton~$\mathcal A$ is called \emph{partial}.
If $|\Sigma| = 1$, we call $\mathcal A$ a \emph{unary} automaton. 
%
The set $L(\mathcal A) = \{\, w \in \Sigma^* \mid \delta(q_0, w) \in F\,\}$ denotes the language
accepted by $\mathcal A$.

A \emph{semi-automaton} is a finite automaton without a specified start state
and with no specified set of final states.
The properties of being \emph{deterministic}, \emph{partial}, and \emph{complete} of semi-automata are defined as for DFA.
When the context is clear, we call both deterministic finite automata and semi-automata simply \emph{automata}.
We call a \emph{deterministic complete semi-automaton a DCSA} and a \emph{partial deterministic finite automaton a PDFA} for short. If we want to add an explicit initial state $r$ and an explicit set of  final states $S$ to a DCSA $\mathcal A$ or change them in a DFA~$\mathcal A$, we use the notation $\mathcal A_{r,S}$.

A \emph{nondeterministic finite automaton (NFA)} $\mathcal A$ is a tuple $\mathcal A = (\Sigma, Q, \delta, s_0, F)$
where $\delta \subseteq Q \times \Sigma \times Q$ is an arbitrary relation. Hence, they generalize
deterministic automata. With a nondeterministic automaton $\mathcal A$
we also associate the set of accepted words $L(\mathcal A) = \{ w \in \Sigma^* \mid \mbox{$w$ labels a path from $s_0$
to some state in $F$} \}$. We refer to~\cite{HopMotUll2001}
for a more formal treatment. In this work, when we only use the word automaton without any adjective, we always
mean a deterministic automaton.

An automaton $\mathcal A$ is called \emph{synchronizing} if there exists a word $w \in \Sigma^*$ with $|\delta(Q, w)| = 1$. In this case, we call $w$ a \emph{synchronizing word} for $\mathcal A$.
For a word $w$, we call a state in $\delta(Q, w)$ an \emph{active} state. 
We call a state $q\in Q$ with $\delta(Q, w)=\{q\}$ for some $w\in \Sigma^*$ a \emph{synchronizing state}.
A state from which some final state is reachable is called \emph{co-accessible}.
For a set $S \subseteq Q$, we say $S$ is \emph{reachable} from $Q$ or $Q$ is synchronizable to $S$ if there exists a word $w \in \Sigma^*$ such that $\delta(Q, w) = S$.
We call an automaton  \emph{initially connected}, if every state is reachable from the start~state.

\begin{fact}\cite{Vol2008} \label{thm:unrestricted_sync_poly_time}
	For any DCSA, we
	can decide if it is synchronizing in polynomial time $O(|\Sigma||Q|^2)$.
	Additionally, we can compute a synchronizing word of length at most $O(|Q|^3)$ in time~$O(|Q|^3 + |Q|^2|\Sigma|))$.
\end{fact}
The following obvious remark will be used frequently without further mentioning.
\begin{lemma}
	\label{lem:append_sync} 
	Let $\mathcal A = (\Sigma, Q, \delta)$ be a DCSA and $w\in \Sigma^*$ be a synchronizing word for $\mathcal A$. Then for every $u, v \in \Sigma^*$, the word $uwv$ is also synchronizing for~$\mathcal A$. 
\end{lemma}


For a fixed PDFA $\mathcal B = (\Sigma, P, \mu, p_0, F)$,
we define the \emph{constrained synchronization problem}:


\begin{decproblem}\label{def:problem_L-constr_Sync}
  \problemtitle{\cite{DBLP:conf/mfcs/FernauGHHVW19}~\textsc{$L(\mathcal B)$-Constr-Sync}}
  \probleminput{Deterministic complete semi-automaton $\mathcal A = (\Sigma, Q, \delta)$.}
  \problemquestion{Is there a synchronizing word $w \in \Sigma^*$ for $\mathcal A$ with  $w \in L(\mathcal B)$?}
\end{decproblem}

The automaton $\mathcal B$ will be called the \emph{constraint automaton}.
If an automaton $\mathcal A$ is a yes-instance of \textsc{$L(\mathcal B)$-Constr-Sync} we call $\mathcal A$ \emph{synchronizing with respect to $\mathcal{B}$}. Occasionally,
we do not specify $\mathcal{B}$ and rather talk about \textsc{$L$-Constr-Sync}.

We assume the reader to have some basic knowledge in computational complexity theory and formal language theory, as contained, e.g., in~\cite{HopMotUll2001}. For instance, we make use of regular expressions to describe languages,
or use many-one polynomial time reductions.
We write $\varepsilon$ for the empty word, and for $w \in \Sigma^*$ we denote by $|w|$
the length of $w$. For some language $L\subseteq \Sigma^*$,
we denote by $\pref(L) = \{ w \mid \exists u \in \Sigma^* : wu \in L \}$,
$\suff(L) = \{ w \mid \exists u \in \Sigma^* : uw \in L \}$
and $\factor(L) = \{ w \mid \exists u,v \in \Sigma^* : uwv \in L \}$
the set of \emph{prefixes}, \emph{suffixes} and \emph{factors} of words in $L$.
The language $L$ is called \emph{prefix-free} if
for each $w \in L$ we have $\pref(w) \cap L = \{w\}$.
If $u, w \in \Sigma^*$, a prefix $u \in \pref(w)$ is called a \emph{proper prefix} if
$u \ne w$.
For $L \subseteq \Sigma^*$ and $u \in\Sigma^*$, the language $u^{-1}L = \{ w \in \Sigma \mid uw \in L \}$
is called a \emph{quotient} (of $L$ by  $u$).
We identify singleton sets with its elements. 
And we make use of complexity classes like $\PTIME$, $\NP$, or $\PSPACE$. 
%
%
%

A \emph{trap} (or \emph{sink}) state in a (semi\nobreakdash-)automaton $\mathcal A = (\Sigma, Q, \delta)$
is a state $q \in Q$ such that $\delta(q, x) = q$ for each $x \in \Sigma$. If a synchronizable
automaton admits a sink state, then this is the only state to which we could synchronize
every other state, as it could only map to itself.
For an automaton $\mathcal A = (\Sigma, Q, \delta, q_0, F)$, 
we say that two states $q, q' \in Q$ are connected, if one is reachable from the other,
i.e., we have a word $u \in \Sigma^*$ such that $\delta(q, u) = q'$. 
A subset $S \subseteq Q$ of states is called
\emph{strongly connected}, if all pairs from $S$ are connected.
A maximal strongly connected subset is called a \emph{strongly connected component}.
By combining Proposition~3.2 and Proposition~5.1 from~\cite{Eilenberg1974}, we get
the next result.

\begin{lemma}\label{lem:unitary}
 For any automaton $\mathcal B = (\Sigma, P, \mu, p_0, F)$
 and any $p \in P$,
 we have $L(\mathcal B_{p, \{p\}}) = C^*$
 for some regular prefix-free set $C \subseteq \Sigma^*$. 
\end{lemma}

We will also need the following combinatorial lemma from~\cite{SchuetzenbergerLyndon62}.

\begin{lemma} \cite{SchuetzenbergerLyndon62}
\label{lem:lyndon-schuetzerger}
 Let $u,v \in \Sigma^*$.
 If $u^m = v^n$ and $m \ge 1$, then $u$ and $v$ are powers of a common word.
\end{lemma}

In \cite{DBLP:journals/jcss/LuksM88,DBLP:journals/cc/BlondinKM16} the decision 
problem \textsc{SetTransporter} was introduced. In general it is $\PSPACE$-complete.

\begin{decproblem}\label{def:problem_set_transporter}
  \problemtitle{\cite{DBLP:journals/jcss/LuksM88,DBLP:journals/cc/BlondinKM16}~\textsc{SetTransporter}}
  \probleminput{DCSA $\mathcal A = (\Sigma, Q, \delta)$ and two subsets $S, T \subseteq Q$.}
  \problemquestion{Is there a word $w \in \Sigma^*$ such that $\delta(S, w) \subseteq T$?}
\end{decproblem}

\noindent We will only use the following variant, which has the same 
complexity.

\begin{decproblem}\label{def:unary_set_transpoer}
  \problemtitle{\textsc{DisjointSetTransporter}}
  \probleminput{DCSA $\mathcal A = (\Sigma, Q, \delta)$ and two subsets $S, T \subseteq Q$ with $S \cap T = \emptyset$.}
  \problemquestion{Is there a word $w \in \Sigma^*$ such that $\delta(S, w) \subseteq T$?}
\end{decproblem}

\begin{restatable}[]{proposition}{disjointequivalentsettransporter}
\label{prop:disjoint_equivalent_set_transporter}
 The problems \textsc{SetTransporter} and \textsc{DisjointSetTransporter} are equivalent under
 polynomial time many-one reductions.
\end{restatable}

We will use Problem~\ref{def:unary_set_transpoer} for unary input DCSAs. 

\begin{restatable}[]{proposition}{settransporternpcomplete}
\label{prop:set_transporter_np_complete}
 For unary DCSAs problem \textsc{SetTransporter}
 is $\NP$-complete.
\end{restatable}

In~\cite{DBLP:conf/mfcs/FernauGHHVW19}, with Theorem~\ref{thm:gen:inNP},
a sufficient criterion was given when the constrained synchronization problem is in $\NP$.

\begin{theorem}
	\label{thm:gen:inNP}
		Let $\mathcal{B} = (\Sigma, P, \mu, p_0, F)$ be a PDFA. 
	Then,  $L(\mathcal B)\textsc{-Constr-Sync}\in\NP$ if  there is a $\sigma\in \Sigma$ such that for all states $p\in P$, if $L(\mathcal{B}_{p,\{p\}})$ is infinite, then  $L(\mathcal{B}_{p,\{p\}})\subseteq \{\sigma\}^*$.
\end{theorem}

\section{Results}

First, in Section~\ref{sec:np_criterion}, we introduce polycyclic automata and generalize Theorem~\ref{thm:gen:inNP}, thus widening the class for which the problem is contained in $\NP$. Then, in Section~\ref{sec:form}, we take a closer look at 
polycyclic automata. We determine their form, show that they admit definitions by partial and by nondeterministic
automata and prove some closure properties.
In Section~\ref{sec:poly_case} we state a general criterion that
gives a polynomial time solvable problem. Then, in Section~\ref{sec:np_case},
we give a sufficient criterion for constraint languages that
give $\NP$-complete problems, which could be used to construct
polycyclic constraint languages that give $\NP$-complete problems.

\subsection{A Sufficient Criterion for Containment in $\NP$}
\label{sec:np_criterion}

The main result of this section is Theorem~\ref{thm:gen:inNP2}.
But first, let us introduce the class of polycyclic partial automata.

\begin{definition}[polycyclic PDFA]
\label{def:polycyclic_automata}
 A PDFA $\mathcal{B} = (\Sigma, P, \mu, p_0, F)$
 is called \emph{polycyclic},
 if for all states $p\in P$
 we have $L(\mathcal{B}_{p,\{p\}})\subseteq \{u_p\}^*$ for some $u_p \in \Sigma^*$.
\end{definition}

The results from Section~\ref{sec:form} will give some justification why we call
these automata polycyclic. In Definition~\ref{def:polycyclic_automata}, languages 
that are given by automata with a single final state, 
which equals the start state, occur.
Our first Lemma~\ref{lem:power_of_up} determines the form of these languages, under the restriction
in question, more precisely.
Note that in any PDFA $\mathcal B = (\Sigma, P, \mu, p_0, F)$
we have either that $L(\mathcal B_{p, \{p\}})$ is infinite 
or $L(\mathcal B_{p, \{p\}}) = \{\varepsilon\}$.

\begin{restatable}[]{lemma}{powerofup}
\label{lem:power_of_up}
 Let $\mathcal{B} = (\Sigma, P, \mu, p_0, F)$ be a PDFA.
 Suppose $p\in P$ such
 that $L(\mathcal{B}_{p,\{p\}})\subseteq \{u_p\}^*$ for some $u_p \in \Sigma^*$.
 Then $L(\mathcal{B}_{p,\{p\}}) = \{u_p^n\}^*$
 for some~$n \ge 1$.
\end{restatable}

Now, we are ready to state the main result of this section.

\begin{restatable}[]{theorem}{thmgeninNP}
\label{thm:gen:inNP2}
 
 
 
 Let $\mathcal{B} = (\Sigma, P, \mu, p_0, F)$ be a polycyclic partial automaton.
 Then $L(\mathcal B)\textsc{-Constr-Sync} \in \NP$.
\end{restatable}
\begin{proof} By Lemma~\ref{lem:power_of_up}, we can assume
$L(\mathcal{B}_{p,\{p\}}) = \{u_p\}^*$ for some $u_p \ne \varepsilon$ for each $p \in P$
such that $L(\mathcal{B}_{p,\{p\}}) \ne \{\varepsilon\}$.
Let $U = \{ u_p \mid L(\mathcal{B}_{p,\{p\}}) = \{u_p\}^* \mbox{ with } u_p \ne \varepsilon \mbox{ for some } p \in P \}$ 
be all such words.
Set $m = |P|$. Every word of length greater than $m-1$
must traverse some cycle. 
Therefore, any word $w \in L(\mathcal B)$
can be partitioned into at most $2m - 1$ substrings
$w = u_{p_1}^{n_1} v_1 \cdots u_{p_{m-1}}^{n_{m-1}}v_{m-1} u_{p_m}^{n_m}$
for numbers $n_{1}, \ldots, n_{m} \ge 0$, $p_{1}, \ldots, p_{m} \in P$
and words $v_1, \ldots, v_{m-1}$. Note that $|v_i| \le m - 1$
for all $i \le m - 1$.
Let $\mathcal A = (\Sigma, Q, \delta)$ be a yes-instance
of $L(\mathcal B)\textsc{-Constr-Sync}$.
Let $w \in L(\mathcal{B})$ be a synchronizing word for $\mathcal A$ partitioned as mentioned above.

\medskip 

\noindent
\underline{Claim~1}: If for some $i \le m$ we have
$n_{i} \ge 2^{|Q|}$, then we can replace it by some $n_i' < 2^{|Q|}$, 
yielding a word $w' \in L(\mathcal B)$ that synchronizes $\mathcal A$. 
This could be seen by considering the non-empty subsets
$$
 \delta(Q, u_{p_1}^{n_1} v_1 \cdots u_{p_{j-1}}^{n_{j-1}} v_{j-1} u_{p_j}^{k})
$$
for $k = 0, 1, \ldots, n_i$. If $n_{i} \ge 2^{|Q|}$, then some such subsets
appears at least twice, but then we can delete the power of $u_{p_i}$
between those appearances.

\medskip

	We will now show that we can decide whether $\mathcal A$ is synchronizing with respect to $\mathcal{B}$ in polynomial time using nondeterminism despite the fact that an actual synchronizing word might be exponentially large. This problem is circumvented by some preprocessing based on modulo arithmetic, and by using
	a more compact representation for a synchronizing word. 
	We will assume we have some numbering of the states, hence the $p_i$
	are numbers. Then, instead of the above form, we will represent a synchronizing word
	in the form
	$w_{code} = 1^{p_1}\#\bin(n_1)v_1 1^{p_2}\#\bin(n_2)v_2 \dots v_{m-1}1^{p_m}\#\bin(n_{m})$, where $\#$
	is some new symbol that works as a separator, and similarly $\{0,1\}\cap \Sigma = \emptyset$
	are new symbols to write down the binary number, or the unary presentation
	of $p_i$, indicating which word $u_{p_i}$ is to be repeated.
	As $\bin(n_i) \le |Q|$ by the above claim and $m$ 
	is fixed by the problem specification,
	the length of $w_{code}$ is polynomially bounded, and we use nondeterminism
	to guess such a code for a synchronizing word.

\medskip 

	\noindent
	\underline{Claim~2}: For each $q\in Q$ and $u \in \Sigma^*$, one can compute in polynomial time numbers $\ell(q),\tau(q)\leq |Q|$ such that, given some number $x$ in binary, based on $\ell(q),\tau(q)$, one can compute in polynomial time a number $y\le |Q|$ such that $\delta(q,u^x)=\delta(q,u^y)$.
\begin{quote}	
		\begin{proof}[Proof of Claim~2 of Theorem~\ref{thm:gen:inNP}]
	For each state $q \in Q$ and $u \in \Sigma^*$, we calculate its $u$-\emph{orbit} $\Orb_{u}(q)$, that is, the set $$\Orb_u(q)=\{q,\delta(q, u), \delta(q,u^2),\dots, \delta(q, u^\tau), \delta(q, u^{\tau+1}),\dots, \delta(q,u^{\tau+\ell-1})\}$$ such that all states in $\Orb_u(q)$ are distinct but $\delta(q,u^{\tau+\ell})=\delta(q,u^\tau)$. Let $\tau(q):=\tau$ and $\ell(q):=\ell$ be the lengths of the tail and the cycle, respectively; these are nonnegative integers that do not exceed $|Q|$. Observe that $\Orb_u(q)$ includes the cycle $\{\delta(q,u^\tau),\dots, \delta(q,u^{\tau+\ell-1})\}$.
We can use this information to calculate $\delta(q,u^x)$, given a nonnegative integer $x$ and a state $q\in Q$, as follows:
(a) If $x\le\tau(q)$, we can find $\delta(q,u^x)\in \Orb_\sigma(q)$. (b) If $x>\tau(q)$, then $\delta(q,u^x)$ lies on the cycle. Compute $y:=\tau(q)+(x-\tau(q)) \pmod {\ell(q)}$. Clearly,  $\delta(q,u^x)= \delta(q,u^y)\in \Orb_\sigma(q)$. The crucial observation is that this computation can be done in time polynomial in $|Q|$ and in $|\bin(x)|$. 
As a consequence, given $S\subseteq Q$ and $x\geq 0$ (in binary), we can compute $\delta(S,u^x)$ in polynomial time.
\end{proof}	
\end{quote}
\medskip

   	The \NP-machine  guesses $w_{code}$ part-by-part, keeping track of the set $S$ of active states of~$\mathcal A$ and of the current state $p$ of $\mathcal B$. Initially,  $S=Q$ and $p=p_0$. 	
	For $i \in \{1,\ldots, m\}$, when guessing the number $n_{i}$ in binary,  by Claim~1 we guess $\log(n_{i})\leq n$ 
    many bits. 
	By Claim~2, we can update  $S:= \delta(S,u_{p_i}^{n_i})$ and $p:=\mu(p,u_{p_{i}}^{n_{i}})$ in polynomial time. 
	%
	After guessing $v_i$, we can simply update $S:= \delta(S,v_i)$ and $p:=\mu(p,v_i)$ by simulating this input, as $|v_i|\leq m=|P|$, which is a constant in our setting.
	Finally, 
	check if $|S|=1$ and 
	if $p\in F$. \qed
\end{proof} 

Comparing Theorem~\ref{thm:gen:inNP2} with Theorem~\ref{thm:gen:inNP}
shows that our generalization allows entire words as a restriction
instead of powers of a single letter for languages of the form $L(\mathcal B_{p,\{p\}})$,
and these words could be different for each state.

\subsection{Properties of Polycyclic Automata}
\label{sec:form}


Here, we look closer at polycyclic automata. 
We find that every 
strongly connected component of a polycyclic PDFA
essentially consists of a single cycle, 
i.e, for each strongly connected component $S \subseteq P$ 
and $p \in S$ we have 
$|\{ x \mid \mbox{$x \in \Sigma$ and $\mu(p,x)$ is defined and in $S$} \}| \le 1$.
Hence, these automata admit a notable simple structure.
We then introduce the class of polycyclic languages. In Proposition~\ref{prop:polycyclic_nea}
we show that these languages could be characterized 
with accepting nondeterministic automata. This result yields
closure under union.


\begin{restatable}[]{proposition}{form}
\label{prop:form}
  Let $\mathcal{B} = (\Sigma, P, \mu, p_0, F)$ be a PDFA.
  Then every strongly connected component of $\mathcal B$ 
  is a single cycle if and only if $\mathcal B$ is polycyclic.
\end{restatable}

We transfer our definition from automata to languages.

\begin{definition}\label{def:polycyclic-lang}
 A language $L \subseteq \Sigma^*$
 is called \emph{polycyclic}, 
 if there exists a polycylic PDFA accepting it.
\end{definition}

Hence, we have the result that the constrained synchronization
problem is in $\NP$ if the constraint language is polycyclic.
By our results, if $L$ gives a constrained synchronization problem outside of \NP,
then $L$ could not be polycyclic. But we also state a simpler necessary criterion.

\begin{restatable}[]{lemma}{notpolycyclic}
\label{lem:not_polycyclic}
 Let $L \subseteq \Sigma^*$ and let $a,b \in \Sigma$ be distinct letters.
  If we find $u \in \Sigma^*$ and $a,b \in \Sigma^+$
  such that $u(a+b)^* \subseteq L$, then $L$ is not polycyclic.
\end{restatable}

By adding a trap state, we can convert every PDFA into a complete DFA accepting the same language.
But the resulting complete DFA is not polycyclic anymore for $|\Sigma| > 1$, as the trap
state has a cycle for every letter. The language $L = ab^*$ is polycyclic,
but its complement $b(a+b)^* \cup a(a+b)^*a(a+b)^*$ is not polycyclic 
by Lemma~\ref{lem:not_polycyclic}. Hence, the polycyclic languages
are not closed under complement, which implies that we could not have a structural characterization
in terms of complete DFA without reference to the set of final states.
However, we can use nondeterministic automata in the definition of polycyclic languages. 
We need the next lemma
to prove this claim.

%
\begin{restatable}[]{lemma}{unitarynotunion}
\label{lem:unitary_not_union}
 Let $\mathcal B = (\Sigma, P, \mu, p_0, F)$ be a PDFA
 such that for some state $p\in P$
 we have $L(\mathcal{B}_{p,\{p\}})\subseteq v^* \cup w^*$.
 Then
 $L(\mathcal B_{p, \{p\}}) \subseteq u^*$ for some word $u \in \Sigma^*$.
\end{restatable}

With Lemma~\ref{lem:unitary_not_union} we can prove the next characterization
by NFAs.

\begin{restatable}[]{proposition}{polycyclicnea}
\label{prop:polycyclic_nea}
  A language $L \subseteq \Sigma^*$ is polycyclic if and only if it is accepted by a nondeterministic
  automaton $\mathcal A = (\Sigma, Q, \delta, s_0, F)$
  such that for all states $p\in Q$
  we have $L(\mathcal{A}_{p,\{p\}})\subseteq \{u_p\}^*$ for some $u_p \in \Sigma^*$.
\end{restatable}

A useful property, which will be used in Section~\ref{sec:np_case}
for constructing examples that yield $\NP$-complete problems,
is that the class of polycyclic languages is closed under concatenation.
We need the next lemma to prove this claim, which gives a certain normal form.

\begin{restatable}[]{lemma}{nocyclestart}
\label{lem:no_cycle_start}
 Let $L \subseteq \Sigma^*$ be a polycyclic language.
 Then, there exists an accepting polycyclic PDFA $\mathcal B = (\Sigma, P, \mu, p_0, F)$
 such that $p_0$ is not contained in any cycle, i.e., $L(\mathcal B_{p_0,\{p_0\}}) = \{\varepsilon\}$.
\end{restatable}

Intuitively, for an automaton $\mathcal B = (\Sigma, P, \mu, p_0, F)$ that has the form as stated in Lemma~\ref{lem:no_cycle_start}, we
can compute its concatenation $L \cdot L(\mathcal B)$ 
with another regular language $L \subseteq \Sigma^*$ by identifying $p_0$
with every final state of an automaton for $L$.

\begin{restatable}[]{proposition}{polycyclicconcat}
\label{prop:polycyclic_concat}
If $U, V \subseteq \Sigma^*$ are polycyclic, then $U\cdot V$ is polycyclic.
\end{restatable}

We also have further closure properties. 

\begin{restatable}[]{proposition}{polycyclicclosure}
\label{prop:polycyclic_closure}
 The polycyclic languages are closed under union and  quotients. 
\end{restatable}





 Without proof, we note that polycyclic automata are a special case of 
 solvable automata as introduced in \cite{Rys97}. 
 Solvable automata are constructed out of commutative automata,
 and here polycyclic automata are constructed out of cycles in the same manner\footnote{Solvable automata 
 according to Rystsov~\cite{Rys97} always have a trap state and are complete. If our partial automata
 are not complete, then we can make them complete by adding a trap state and the analogy is meant in this ways, where
 special attention has to be paid to the trap state as it is in general not a single cycle. If the polycyclic
 automaton happens to be complete, Rystsov's~\cite{Rys97} definition has to be altered slightly by not
 demanding the lowest automaton in a composition chain to be a single state complete automaton.}.
 Without getting to technical, let us note that 
 in abstract algebra and the theory of groups,
 a polycyclic group is a group constructed out of cyclic groups in the same manner
 as a solvable group is constructed out of commutative groups~\cite{Robinson:1995}.
 Hence, the naming supports the analogy to group theory quite well.
 Also, let us note that polycyclic automata have cycle rank~\cite{Egg63}
 at most one, hence they have star height at most one.
 But they are properly contained in the languages of star height one,
 as shown for example by $(a+b)^*$.

\subsection{Polynomial Time Solvable Cases} 
\label{sec:poly_case}

\begin{wrapfigure}[8]{r}{0.5\textwidth} 
\centering
\scalebox{.7}{
    \begin{tikzpicture}[>=latex',shorten >=1pt,node distance=2cm,on grid,auto]
     \tikzset{every state/.style={minimum size=1pt}}
\node[state, initial] (q1) {};
\node[state, right of=q1] (q2) {};
\node[state, right of=q2] (q3) {};
\node[state, above left=of q3] (q4) {};
\node[state, right=of q4] (q5) {};
\node[state, accepting, right of=q3] (q6) {};

\path[->] (q1) edge node {$b$} (q2);
\path[->] (q2) edge node {$a$} (q3);
\path[->] (q3) edge node {$a$} (q4);
\path[->] (q4) edge node {$b$} (q5);
\path[->] (q5) edge node {$a$} (q3);
\path[->] (q3) edge node {$b$} (q6);
\end{tikzpicture}}
 \caption{\footnotesize An example constraint automaton with $L(\mathcal B)\textsc{-Constr-Sync} \in \PTIME$.}
   \label{fig:poly_case}
\end{wrapfigure}
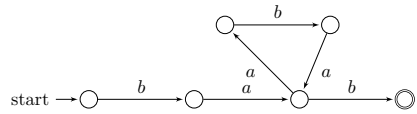

Here, with Proposition~\ref{prop:NP_in_P}, we state a sufficient criterion for a polycyclic constraint
automaton that
gives constrained synchronization problems that are solvable in polynomial time.
Please see Figure~\ref{fig:poly_case} 
for an example constraint automaton whose constrained synchronization problem is in~$\PTIME$
according to Proposition~\ref{prop:NP_in_P}.

\begin{proposition}
\label{prop:NP_in_P}
  Let $\mathcal{B} = (\Sigma, P, \mu, p_0, F)$ be a polycyclic PDFA.
  If for any reachable $p \in P$ with $L(\mathcal B_{p, \{p\}}) \ne \{\varepsilon\}$ 
  we have $L(\mathcal B_{p_0, \{p\}}) \subseteq \suff(L(\mathcal B_{p, \{p\}}))$,
  then the problem $L(\mathcal B)\textsc{-Constr-Sync}$ is solvable
  in polynomial time.
\end{proposition}
\begin{proof}
By Lemma~\ref{lem:power_of_up}, we can assume
$L(\mathcal{B}_{p,\{p\}}) = \{u_p\}^*$ for some $u_p \ne \varepsilon$ for each $p \in P$
such that $L(\mathcal{B}_{p,\{p\}}) \ne \{\varepsilon\}$.
Let $U = \{ u_p \mid L(\mathcal{B}_{p,\{p\}}) = \{u_p\}^* \mbox{ with } u_p \ne \varepsilon \mbox{ for some } p \in P \}$ 
be all such words.
Set $m = |P|$. Every word of length greater than $m-1$
must traverse some cycle. 
Therefore, any word $w \in L(\mathcal B)$
can be partitioned into at most $2m - 1$ substrings
$w = u_{p_1}^{n_1} v_1 \cdots  u_{p_{m-1}}^{n_{m-1}} v_{m-1} u_{p_m}^{n_m}$
for numbers $n_{1}, \ldots, n_{m} \ge 0$, $p_1,\ldots, p_m \in P$ 
and words $v_1, \ldots, v_{m-1}$. Note that $|v_i| \le m - 1$
for all $i \le m-1$.
Let $\mathcal A = (\Sigma, Q, \delta)$ be a yes-instance
of $L(\mathcal B)\textsc{-Constr-Sync}$.
Let $w \in L(\mathcal{B})$ be a synchronizing word for $\mathcal A$ partitioned as mentioned above.
We show that by our assumptions we could choose
the numbers $n_1, \ldots, n_m$ to be strictly smaller than $|Q|$.


\medskip 

\noindent
\underline{Claim~1}: If for some $i \le m$ we have
$n_{i} \ge |Q|$, then we can replace it by $n_i' = |Q| - 1$, 
yielding a word $w' \in L(\mathcal B)$ that synchronizes $\mathcal A$. 
\begin{quote}
    \begin{proof}[Proof of Claim~1 of Proposition~\ref{prop:NP_in_P}]
    Let $j \in \{1,\ldots, m\}$ be arbitrary with $n_j \ge |Q|$ and $u_{p_j} \ne \varepsilon$ (otherwise
    we have nothing to prove).
    Set $u = u_{p_1}^{n_1} v_1 u_{p_2}^{n_2} v_2 \cdots u_{p_{j-1}}^{n_{j-1}} v_{j-1}$ 
    and $S = \delta(Q, u)$.
    By choice of the decomposition of $w$, if $n_j > 0$, then $\mu(p_0, u) = p_j$  
    and $L(\mathcal B_{p_j, \{p_j\}}) = \{ u_{p_j} \}^*$.
    Write $p = p_j$.
    As by assumption $u$ appears as a suffix of some word from $L(\mathcal B_{p, \{p\}})$, we have $u_p^l = vu$
    for some $v \in \Sigma^*$ and $l \ge 0$.
    Hence, for each $k \in \mathbb N_0$ we have $\delta(S, u_p^{lk}) \subseteq S$
    as every word that has $u$ as a suffix maps any state to a state in $S$.
    Let us assume $l = 1$ in the following argument, as this is a fixed parameter
    of $L(\mathcal B)\textsc{-Constr-Sync}$ only depending on $u$. Hence, the conclusion would be the same
    if we replace $u_p$ by $u_p^l$ in the following arguments.

    First, we show $\delta(S, u_p^{|S|}) = \delta(S, u_p^{|S|-1})$.
    For $q \in S$ we have $\delta(q, u_p^{|S|}) = \delta(\delta(q,u_p), u_p^{|S|-1})$ and
    as $\delta(q, u_p) \in S$ this gives $\delta(S, u_p^{|S|}) \subseteq \delta(S, u_p^{|S|-1})$.
    Now let us show the other inclusion $\delta(S, u_p^{|S|-1}) \subseteq \delta(S, u_p^{|S|})$.
    Let $q \in S$. By the pigeonhole principle
    $$
       \delta(q, u_p^{|S|}) \in \{ q, \delta(q, u_p), \ldots, \delta(q, u_p^{|S|-1}) \}.
    $$
    Hence $\delta(q, u_p^{|S|})$ equals $\delta(q, u_p^k)$
    for some $0 \le k < |S|$.
    Choose $a,b \ge 0$ such that $|S| = a(|S| - k) + b$ with $0 \le b < |S| - k$.
    Note that $\delta(q, u_p^{k + l + a(|S| - k)}) = \delta(q, u_p^{k+l})$
    for each $l \ge 0$ because $\delta(q, u_p^{k + |S| - k}) = \delta(q, u_p^{|S|}) = \delta(q, u_p^{k})$.
    Set $q' = \delta(\delta(q, u_p^{|S|-1}), u_p^{|S| - k - b}).$
    By assumption $q' \in S$.
    Then 
    \begin{align*}
        \delta(q', u_p^{|S|}) & = \delta(q, u_p^{|S|-1 + |S| - k - b + |S|}) \\ 
                              & = \delta(q, u_p^{|S|-1 + |S| - k + a(|S| - k)}) \\
                              & = \delta(q, u_p^{|S|-1}).
    \end{align*}
    So $\delta(q, u_p^{|S|-1}) \in \delta(S, u_p^{|S|})$.
    Hence, regarding our original problem, if $n_j \ge |Q| \ge |S|$,
    we have $\delta(Q, u u_{p_j}^{n_j}) = \delta(Q, u u_{p_j}^{|Q|-1})$,
    as inductively $\delta(Q, u u_{p_j}^{n_j}) = \delta(S, u_{p_j}^{n_j}) = \delta(S, u_{p_j}^{|S|-1})
    =  \delta(S, u_{p_j}^{|Q|-1}) = \delta(Q, uu_{p_j}^{|Q|-1})$. 
    \end{proof}
\end{quote}

\medskip

So, to find out if we have any synchronizing word in $L(\mathcal B)$, we only have to test
the finitely many words
$$
 u_{p_1}^{n_1} v_1 \cdots u_{p_{m-1}}^{n_{m-1}} v_{m-1} u_{p_m}^{n_m} 
$$
for  $n_1, \ldots, n_{m-1}, n_m \in \{0,1,\ldots, |Q|-1\}$, $u_{p_1}, \ldots, u_{p_{m-1}}, u_{p_m} \in U_p$ and words $v_1,\ldots, v_{m-1}$ of length at most $m$. As $m = |P|$ and $U_p$
is fixed, we have to test $O(|Q|^m)$ many words. For each word
$w =u_{p_1}^{n_1} v_1 \cdots u_{p_{m-1}}^{n_{m-1}} v_{m-1} u_{p_m}^{n_m}  $, we have to read in this word starting from 
any state in $Q$ and check if a unique state results, i.e., check if $\delta(q, w) = \delta(q', w)$
for $q, q' \in Q$.
All these operations could be performed in polynomial time with parameter $|Q|$.~$\qed$
\end{proof}

\subsection{$\NP$-complete Cases}
\label{sec:np_case}

In~\cite{DBLP:conf/mfcs/FernauGHHVW19} it was shown that 
for the constraint language $L = ba^*b$ and for the languages $L_i = (b^*a)^i$
with $i \ge 2$ the corresponding constrained synchronization problems
are $\NP$-complete.
All $\NP$-complete problems with a $3$-state constraint automaton and a binary alphabet
where determined in~\cite{DBLP:journals/corr/abs-2005-05907}. Here, with Proposition~\ref{prop:NPc},
we state a general scheme, involving the concatenation operator, to construct $\NP$-hard problems. 
As, by Proposition~\ref{prop:polycyclic_concat}, the polycyclic languages are closed under concatenation,
this gives us a method to construct $\NP$-complete constrained synchronization
problems with polycyclic constraint languages.

\begin{proposition}
\label{prop:NPc}
 Suppose we find $u, v \in \Sigma^*$ 
 such that we can write
$
 L = u v^* U
$
 for some non-empty language $U \subseteq \Sigma^*$
 with 
 $$
  u \notin \factor(v^*), \quad
  v \notin \factor(U), \quad
  \pref(v^*) \cap U = \emptyset.
 $$
 Then $L\textsc{-Constr-Sync}$ is $\NP$-hard.
\end{proposition} 
\begin{proof}
 Note that $u\notin \factor(v^*)$ implies $u \ne \varepsilon$,
 $v \notin \factor(U)$ implies $v \ne \varepsilon$ and $\pref(v^*) \cap U = \emptyset$
 with $U \ne \emptyset$ implies $U \cap \Sigma^+ \ne \emptyset$.
 We show $\NP$-hardness by reduction from $\textsc{DisjointSetTransporter}$
 for unary automata, which is $\NP$-complete by Proposition~\ref{prop:disjoint_equivalent_set_transporter}
 and Proposition~\ref{prop:set_transporter_np_complete}.
 Let $(\mathcal A, S, T)$ be an instance of $\textsc{DisjointSetTransporter}$
 with unary semi-automaton $\mathcal A = (\{c\}, Q, \delta)$.
 Write $v^{|u|} = x_1 \cdots x_n$ with $x_i \in \Sigma$ for $i \in \{1,\ldots, n\}$.
 We construct a new semi-automaton $\mathcal A' = (\Sigma, Q', \delta')$
 with $Q' = Q \cup Q_1 \cup \ldots \cup Q_{n-1} \cup \{t\}$,
 where $Q_i = \{ q_i \mid q \in Q \}$ are disjoint copies of $Q$
 and $t$ is a new state that will work as a trap state in $\mathcal A'$.
 Assume $\varphi_i : Q \to Q_i$
 for $i \in \{1,\ldots, n-1\}$ 
 are bijections with $\varphi_i(q) = q_i$.
 Also, to simplify the formulas, set $Q = Q_0$ and $\varphi_0 : Q \to Q$ the identity map.
 Choose some $\hat s \in S$. Then, for $r \in Q'$ and $x \in \Sigma$ we define
 $$
  \delta'(r, x) = \left\{\begin{array}{ll}
     \varphi_{i+1}(q)      &  \quad \exists i \in \{0, 1,\ldots, n-2\} : r \in Q_i, r = \varphi_i(q), x = x_{i+1}; \\
     \delta(q,c)           &  \quad r \in Q_{n-1}, r = \varphi_{n-1}(q), x = x_{n}; \\
     \hat s                &  \quad \exists i \in \{0, 1,\ldots, n-1 \} : r \in Q_i, r = \varphi_i(q), q \notin T \cup S, x \ne x_{i+1}; \\
     q                     &  \quad \exists i \in \{0, 1,\ldots, n-1 \} : r \in Q_i, r = \varphi_i(q), q \in S, x \ne x_{i+1}; \\
     t                     &  \quad \exists i \in \{0, 1,\ldots, n-1 \} : r \in Q_i, r = \varphi_i(q), q \in T, x \ne x_{i+1}; \\
     t                     &  \quad r = t.
  \end{array}\right.
 $$
 Note that by construction of $\mathcal A'$, we have
 for $q \in Q \cup Q_1 \cup \ldots \cup Q_{n-1}$ and $w \in \Sigma^*$
 \begin{equation}
  \delta(q, w) = t \Leftrightarrow \exists x,y,z \in \Sigma^* : w = xyz, \delta(q, x) \in T, y \notin \pref(v^{|u|})
 \end{equation}
 and for $q, q' \in Q$
 \begin{equation}\label{eqn:Astartransition}
   \delta(q, c) = q' \mbox{ in $\mathcal A$ } \Leftrightarrow 
   \delta'(q, v^{|u|}) = q' \mbox{ in $\mathcal A'$ }.
 \end{equation}
 \begin{enumerate}
 \item Suppose we have $c^m$ such that $\delta(S, c^m) \subseteq T$ in $\mathcal A$.
   Because $u$ is not a factor of $v^{2|u|}$, by construction of $\mathcal A'$, we have
   $\delta'(Q' \setminus ( T \cup \{t\} ), u) = S$, where $S$
   is reached as $|u| \le |v^{|u|}|$.
   This yields $\delta'(Q \setminus (T \cup \{t\}), ua^m) \subseteq T$.
   By construction of $\mathcal A'$,
   $\delta'(T, x) = \{t\}$ for any $x \in U \cap \Sigma^+$, as $x \notin \pref(v^*)$ and $v$ is not a factor of $x$.
   Note that we need the last condition to ensure that we do not do a transition from a state in $Q$
   to a state in $Q$ in $\mathcal A'$, for if a word does this, it must have $v^{|u|}$ as a prefix.
   Also $\delta'(T, u) = \{t\}$, as $0 < |u| \le |v^{|u|}|$
   and $u \notin \pref(v^*)$ by the assumption $u \notin \factor(v^*)$.
   
 \item Suppose we have $w \in L(\mathcal B)$ that synchronizes $\mathcal A'$.
  Then, as $t$ is a trap state, $\delta'(Q', w) = \{t\}$.
  Write $uv^mx$ with $x \in U$.
  By construction of $\mathcal A'$, we have, 
  as in the previous case, $\delta'(Q \setminus (T \cup \{t\}), u) = S$.
  Write $m = a|u| + b$ with $0 \le b < |u|$.
  We argue that we must have $\delta'(S, v^{a|u|}) \subseteq T$.
  For assume we have $q \in S$ with $\delta'(q, v^{a|u|}) \notin T$,
  then $\delta'(q, v^{a|u| + b}) \in Q_{b \cdot |v|}$ by construction of $\mathcal A'$.
  As $S \cap T = \emptyset$, and hence $q \notin T$, 
  by construction of $\mathcal A'$, this gives $\delta'(q, v^m x) \in S \cup Q_1 \cup \ldots \cup Q_{|v|-1}$,
  as $x$ is not a prefix of $v$ and does not contain $v$ as a factor. 
  More specifically, to go from $q' \in Q_{b|v|}$ to $Q_{(b+1)|v|}$, or $Q$ in case $b + |v| = n$,
  we have to read $v$. But $x$ does not contain $v$ as a factor.
  By construction of $\mathcal A'$, for any $y \in \Sigma^{|v|} \setminus\{v\}$,
  we have $\delta'(q', y) \in S$ if $q' \in Q_{|v|b}$. 
  This gives that if $x = x'x''x'''$ with $|x''| \le |v|$ and $\delta'(q, v^mx')\in S$
  we have $\delta'(q, v^mx'x'') \in S \cup Q_1 \cup \ldots \cup Q_{|v|-1}$,
  as after reading at most $|v|$ symbols of any factor of $x$,
  starting in a state from $S$, we must return to this state at least once
  while reading this factor.
  Note that by the above reasoning we find a prefix $x'$ of $x$
  with $|x'| \le |v|$ such that $\delta'(q, v^mx') \in S$ in case $|x| \ge |v|$.
  If $|x| < |v|$, then either $\delta'(q', x) \in Q_{|b|+|x|}$
  or $\delta'(q', x) \in S$.
  So, in no case could we end up in the state $t$.
  Hence, we must have $\delta'(q, v^{a|u|}) \in T$ for each $q \in S$.
  By Equation~\eqref{eqn:Astartransition}
  we get $\delta(S, c^a) \subseteq T$.
 \end{enumerate}
 So, we have a synchronizing word for $\mathcal A'$ from $L(\mathcal B)$
 if and only if we can map the set $S$ into $T$ in $\mathcal A$. $\qed$
\end{proof}



If $u,v \in \Sigma^*$ with $u \notin \factor(v^*)$, 
by choosing  $U = \{ w \}$ with $w \notin \pref(v) \cup \Sigma^* v \Sigma^*$
we get that $uv^*w$ gives an $\NP$-complete problem.
Also our result shows that for example $aa(ba)^*aaa^*a$ 
yields an $\NP$-complete problem.


\section{Conclusion}
\label{sec:conclusion}

We introduced the class of polycyclic automata 
and showed that for polycyclic constraint automata,
the constrained synchronization
problem is in $\NP$.
For these contraint automata, we have given a sufficient criterion that yields problems
in $\PTIME$, and a criterion that yields problems that are $\NP$-complete.
However, both criteria do not cover all cases. Hence, there are still polycyclic constraint automata
left for which we do not know the exact
computational complexity in $\NP$ of the constrained synchronization problem.
A dichotomy theorem for our class, i.e, that every
problem is either $\NP$-complete or in $\PTIME$, would be very interesting.
However, much more interesting would be if we could find any candidate $\NP$-intermediate problems.
Lastly, we took a closer look at polycyclic automata, determined their form and also gave a characterization
in terms of nondeterministic automata. We also introduced polycyclic languages and proved
basic closure properties for this class.

{\footnotesize
\medskip \noindent
\textbf{Acknowledgement}. 
I thank Prof. Dr.
Mikhail V. Volkov for suggesting the problem of constrained synchronization
during the workshop `Modern Complexity Aspects of Formal Languages' that took place  at Trier University  11.--15.\ February, 2019.
The financial support of this workshop by the DFG-funded project FE560/9-1 is gratefully acknowledged.
I also thank the anonymous reviewers for their suggestions which helped in improving the presentation.
}

\bibliographystyle{splncs04}
\bibliography{ms}

\clearpage
\section{Appendix}

Here, we collect some proofs not given in the main text. 

\subsection{Proof of Proposition~\ref{prop:disjoint_equivalent_set_transporter} (See page~\pageref{prop:disjoint_equivalent_set_transporter})}
\disjointequivalentsettransporter*
\begin{proof} 
 Let $(\mathcal A, S, T)$ be an instance of \textsc{DisjointSetTransporter},
 then we can feed it unaltered into an algorithm for \textsc{SetTransporter}.
 Now assume $(\mathcal A, S, T)$ is an instance of \textsc{SetTransporter}.
 If $S \subseteq T$, then the empty word maps $S$ into $T$ and
 this case can be easily checked.
 Hence, assume $S \nsubseteq T$. In this case, any word $w$
 with $\delta(S, w) \subseteq T$ must be non-empty.
 Let $S' = \{ s' \mid s \in S \}$ be a disjoint copy of $S$ with $S' \cap Q = \emptyset$. Construct 
 $\mathcal A' = (\Sigma, Q', \delta')$ with $Q' = Q \cup S'$ and
 $\delta'(q, x) = \delta(q, x)$ for $q \in Q$ and $x \in \Sigma$,
 and $\delta'(s', x) = \delta(s, x)$ for $s \in S$ and $x \in \Sigma$.
 We claim that $w \in \Sigma^+$ maps $S$ into $T$ in $\mathcal A$
 if and only if it maps $S'$ into $T$ in $\mathcal A'$.
 By construction of $\mathcal A'$, if $\delta(S, w) \subseteq T$
 with $w \ne \varepsilon$, then $\delta'(S', w) = \delta'(S, w) = \delta(S, w) \subseteq T$.
 Conversely, if $\delta'(S', w) \subseteq T$ with $w \ne \varepsilon$.
 As $\delta(s, w) = \delta'(s', w)$ for every non-empty $w$ and $s \in S$,
 this yields $\delta(S, w) \subseteq T$. 
 So, we can solve $(\mathcal A, S, T)$ by solving the instance $(\mathcal A', S', T)$
 of \textsc{DisjointSetTransporter}. $\qed$
\end{proof}

\subsection{Proof of Proposition~\ref{prop:set_transporter_np_complete} (See page~\pageref{prop:set_transporter_np_complete})}
\settransporternpcomplete*
\begin{proof} 
 We will use the next problem from \cite{Koz77}, which is $\PSPACE$-complete in general, but $\NP$-complete
 for unary automata, see \cite{fernau2017problems}.

\begin{decproblem}\label{def:problem_AutInt}
  \problemtitle{\textsc{Intersection-Non-Emptiness}}
  \probleminput{Deterministic complete automata $\mathcal A_1$, $\mathcal A_2$, \ldots, $\mathcal A_k$.}
  \problemquestion{Is there a word accepted by them all?}
\end{decproblem}

 First, we will show that the problems \textsc{SetTranspoter}
 and \textsc{Intersection-Non-Emptiness} are equivalent under polynomial time many-one
 reductions.
 Let $\mathcal A_1,\ldots, \mathcal A_k$
 with $\mathcal A_i = (\Sigma, Q_i, \delta_i, s_i, F_i)$ for $i \in \{1,\ldots, k\}$
 be an instance
 of \textsc{Intersection-Non-Emptiness}. Assume $Q_i \cap Q_j = \emptyset$
 for $i \ne j$, otherwise we can replace the state sets by disjoint copies.
 We can construct
 the semi-automaton $\mathcal A = (\Sigma, Q, \delta)$
 with $Q = Q_1 \cup \ldots \cup Q_k$ and $\delta(q, x) = \delta_i(q, x)$
 for $q \in Q_i$. Set $S = \{ s_1, \ldots, s_k \}$
 and $T = F_1 \cup \ldots \cup F_k$.
 Then we have a word $w \in \Sigma^*$ with $\delta(S, w) \subseteq T$
 if and only if $\delta_i(s_i, w) \in F_i$ for all $i \in \{1,\ldots, k\}$.
 Conversely, let $(\mathcal A, S, T)$ be an instance
 of \textsc{SetTranspoter}. Suppose $S = \{s_1, \ldots, s_k\}$ with $k = |S|$.
 Set $\mathcal A_i = (\Sigma, Q, \delta, s_i, T)$.
 Then we have a word $w \in \Sigma^*$ with $\delta(s_i, w) \in T$
 for all $i \in \{1,\ldots, k\}$ if and only if $\delta(S, w) \subseteq T$.

 As for $|\Sigma| = 1$ the problem \textsc{Intersection-Non-Emptiness}
 is $\NP$-complete, this implies that for $|\Sigma| = 1$
 also \textsc{SetTransporter} is $\NP$-complete. $\qed $
\end{proof}

\subsection{Proof of Proposition~\ref{prop:form} (See page~\pageref{prop:form})}
\form*
\begin{proof}
 First, suppose every strongly connected component forms a single cycle.
 Let $p \in P$ and assume $L(\mathcal B_{p,\{p\}}) \ne \{\varepsilon\}$.
 By assumption, then for each $q \in S$ we have exactly
 one $x \in \Sigma$ such that $\mu(q, x) \in S$.
  Then, inductively, for $q, q' \in S$
  we have a unique minimal word $w \in \Sigma^+$
  with $\mu(q, w) = q'$.
  Choose $u_p \in \Sigma^+$ minimal
  with $\mu(p, u_p) = p$. 
  We claim $L(\mathcal B_{p,\{p\}}) \subseteq \{u_p\}^*$.
  Assume $\mu(p, w) = p$ with $w \in \Sigma^+$.
  Then write $w = w_1 \cdots w_n$ with $w_i \in \Sigma$ for $i \in \{1,\ldots, n\}$
  and $u_p = u_1 \cdots u_m$ with $u_j \in \Sigma$ for $j \in \{1,\ldots, m\}$.
  We have $n \ge m$.
  By assumption, $\mu(p, w_1 \cdots w_i) = \mu(p, \cdots u_i)$ and $w_i = u_i$ for $i \in \{1,\ldots, m\}$.
  Hence $w_1 \cdots w_m = u_p$.
  Continuing inductively, and by the minimal choice of $|u_p|$,
  we find that $|u_p|$ is a divisor of $|w|$
  and then that $w \in u_p^*$.
 
  Conversely, assume $\mathcal B$ is polycyclic
  and let $S \subseteq Q$ be the states of some strongly connected component.
%
   Assume we have $x,y \in \Sigma$
   and $p \in S$ with such that both $\mu(p,x)$
   and $\mu(p,y)$ are defined and $\{ \mu(p, x), \mu(p, y) \} \subseteq S$.
   As, by assumption, $\mathcal B$ is polycyclic,
   we have $L(\mathcal B_{p,\{p\}}) \subseteq \{u_p\}^*$
   for some $u_p \in \Sigma^*$.
   If $x \ne y$, then by determinism of $\mathcal B$,
   also $\mu(p, x) \ne \mu(p,y)$, and as $S$ is a strongly connected component
   we find $u, u'\in \Sigma^*$
   such that $\mu(p, xu) = p$ and $\mu(p, yu') = p$.
   But $xu, yu' \in u_p^*$ would imply that they both start with the same symbol,
   hence this is not possible and we must have $x = y$.
   Hence, we could have at most one $x \in \Sigma$
   such that $\mu(p, x)$ is defined and in $S$. $\qed$
\end{proof}

\subsection{Proof of Lemma~\ref{lem:power_of_up} (See page~\pageref{lem:power_of_up})}
\powerofup*
\begin{proof}
 By Lemma~\ref{lem:unitary}, we have $L(\mathcal{B}_{p,\{p\}}) = C^*$
 for some regular prefix free set $C \subseteq \Sigma^*$.
 But as $C \subseteq \{u_p\}^*$ is prefix-free, $\{ u_p^n, u_p^m \} \subseteq C$ implies $n = m$, for otherwise
 one word is a prefix of another. 
 Hence $C = \{ u_p^n \}$
 for some $n \ge 1$, \qed
\end{proof}

\subsection{Proof of Lemma~\ref{lem:not_polycyclic} (See page~\pageref{lem:not_polycyclic})}
\notpolycyclic*
\begin{proof}
 Let $\mathcal B = (\Sigma, P, \mu, p_0, F)$ be a PDFA with $L = L(\mathcal B)$.
 Then we find $n_1 > 0$ such that in $\{ \delta(p_0, ux) \mid x \in (a+b)^{n_1} \}$
 two words are mapped to the same state, i.e., $\delta(p_0, ux_1) = \delta(p_0, uy_1)$
 with $x_1 \ne y_1$ and $x_1,y_1 \in (a+b)^{n_1}$. Note that by assumption, states
 of the form $\delta(p_0, ux)$ with $x \in (a+b)^*$ must be defined.
 Denote this state by $p_1$.
 Then, inductively, for $i > 1$, we find $n_i > 0$ such that
 $\delta(p_{i-1}, x_i) = \delta(p_{i-1}, y_i)$ with $x_i,y_i \in (a+b)^{n_i}$ distinct.
 Denote the common state by $p_i$.
 We got a sequence of states. Hence, we have indices $k > j$ 
 such that $p_k = p_j$. 
 Then $\delta(p_j, x) = \delta(p_j,y) = p_k$
 with $x =  x_j x_{j+1} \cdots x_k$
 and $y = y_j x_{j+1} \cdots x_k$, where $x \ne y$ as $x_j \ne y_j$ with $|x| = |y|$.
 So $\{ x,y \}^* \subseteq L(\mathcal B_{p_j, \{ p_j\}})$.
 But a set of
 this form could not be contained in a set of the form $u^*$, 
 as $| \{ x,y \}^m | = 2^m$, but $|u^* \cap \Sigma^n| \le 1$ for all $n \ge 0$.~\qed
\end{proof}

\subsection{Proof of Lemma~\ref{lem:unitary_not_union} (See page~\pageref{lem:unitary_not_union})}
\unitarynotunion*
\begin{proof}
 If $L(\mathcal B_{p,\{p\}}) \cap v^* = \{\varepsilon\}$
 or $L(\mathcal B_{p,\{p\}}) \cap w^* = \{\varepsilon\}$,
 we can set $u = w$ or $u = v$.
 So, assume this is not the case.
 Then we find $x,y \in L(\mathcal B_{p, \{p\}})$
 with $x = v^n$ and $y = w^m$ for $n, m \ge 1$.
 As both words start and end at the same state, we have $xy \in L(\mathcal B_{p,\{p\}})$.
 Hence we either have $v^n w^m = v^k$ for some $k > n$,
 or $v^n w^m = w^k$ for some $k > m$.
 In the first case, this implies $w^m = v^{k - n}$,
 in the second case, $v^n = w^{k-m}$.
 By Lemma~\ref{lem:lyndon-schuetzerger}, this implies
 in either case that $v$ and $w$
 are powers of a common words, hence
 our claim is implied. $\qed$
\end{proof}

\subsection{Proof of Lemma~\ref{lem:no_cycle_start} (See page~\pageref{lem:no_cycle_start})}
\nocyclestart*
\begin{proof}
 Let $L = L(\mathcal B)$
 with a PDFA $\mathcal B = (\Sigma, P, \mu, p_0, F)$ such that
 every strongly connected component forms a single cycle.
 Suppose $p_0$ is contained in a cycle labelled with $u \in \Sigma^+$.
 Write $u = xv$ with $x \in \Sigma$.
 Intuitively, we are going to ''unfold'' the first step from the start state.
 Introduce the new starting state $p_0'$
 and set $\mu'(p_0', x) = \mu(p_0, x)$
 and $\mu'(p, y) = \mu(p, y)$ for $p \in P$, $y \in \Sigma$.
 Note that for $y \ne x$ we either have no transition from $p_0$
 labelled by $y$, or the state $\mu(p_0, y)$ is not in the same strongly connected component
 as $p_0$, for if it would, we would have more than one cycle in this component.
 If $p_0 \in F$, then also alter the final state set to $F' = F \cup \{p_0'\}$,
 otherwise set $F' = F$.
 Define $\mathcal B' = (\Sigma, P \cup \{p_0'\}, \mu', p_0', F')$.
 It is easy to see that $L(\mathcal B) = L(\mathcal B')$ 
 and that in $\mathcal B'$ also every strongly connected component forms
 a single cycle. $\qed$
\end{proof}

\subsection{Proof of Proposition~\ref{prop:polycyclic_nea} (See page~\pageref{prop:polycyclic_nea})}
\polycyclicnea*
\begin{proof}
 If $L \subseteq \Sigma^*$ is polycyclic, we have a polycyclic PDFA accepting
 it. As nondeterministic automata generalize partial automata, one implication
 is implied.
 Conversely, let  $\mathcal A = (\Sigma, Q, \delta, s_0, F)$
 be a nondeterministic automaton with $L = L(\mathcal A)$
 such that for all states $p\in Q$
 we have $L(\mathcal{A}_{p,\{p\}})\subseteq \{u_p\}^*$ for some $u_p \in \Sigma^*$.
 Let $S \subseteq Q$
 and $u \in \Sigma^*$
 with 
 \begin{equation}\label{eqn:cycle_power_set}
  S = \{ s \mid \exists t \in S : (t, u, s) \in \delta\},
 \end{equation}
 i.e., the word $u$ labels a cycle in the power set automaton\footnote{See \cite{HopMotUll2001}
 for the power set construction for conversion of a nondeterministic automaton into
 an equivalent deterministic automaton.}.
 Construct a sequence $s_i \in S$ for $i \in \mathbb N \setminus\{0\}$ 
 by choosing $s_1 \in S$ arbitrarily, and then, inductively, if $s_i$ was chosen, 
 choose some $s_{i+1} \in \{ s \in S \mid (s, u, s_i) \in \delta\}$. Note that 
 $\{ s \in S \mid (s, u, s_i) \in \delta \} \ne \emptyset$
 by Equation~\eqref{eqn:cycle_power_set}.
 As $S$ is finite, we find $i < j$ with $s_i = s_j$.
 But then $(s_i, u^{j-i}, s_i)$ labels a cycle in $\mathcal A$.
 So, by assumption, $u^{j-i} \subseteq \{ u_{s_i} \}^*$ for some word $u_{s_i} \in \Sigma^*$
 that only depends on $s_i$.
 By Lemma~\ref{lem:lyndon-schuetzerger}, both $u$ and $u_{s_i}$
 are powers of a common word.
 Let $U = \{ u_s \mid s \in S, L(\mathcal B_{s,\{s\}}) \subseteq \{u_s\}^*, u_s \ne \varepsilon \}$.
 Then, as $U$ is finite, $V = \{ x \in \Sigma^* \mid u \in U, u \subseteq x^* \}$
 is finite. 
 By the above reasoning, for each $u \in \Sigma^*$
 which fulfills Equation~\eqref{eqn:cycle_power_set},
 we have $u \in \bigcup_{x \in V} x^*$.
 By Lemma~\ref{lem:unitary_not_union} the set of all these
 words is contained in a language of the form $y^*$
 for some $y \in \Sigma^*$. $\qed$
\end{proof}

\subsection{Proof of Proposition~\ref{prop:polycyclic_concat} (See page~\pageref{prop:polycyclic_concat})}
\polycyclicconcat*
\begin{proof}
 Let $\mathcal B = (\Sigma, P, \mu, p_0, F)$
 and $\mathcal B' = (\Sigma, P', \mu', p_0', F')$
 be PDFAs with $U = L(\mathcal B)$
 and $V = L(\mathcal B')$
 such that,  by Lemma~\ref{lem:no_cycle_start}, the start state $p_0'$ is not contained in any cycle. 
 Construct $\mathcal B'' = (\Sigma, P'', \mu'', p_0, F'')$
 by setting $P'' = P \cup ( P' \setminus \{ p_0' \} )$, 
 $$
  F'' = \left\{ \begin{array}{ll}
   F' & \mbox{ if } p_0' \notin F' \\ 
   F \cup F' & \mbox{ otherwise,}
  \end{array}\right.
 $$
 and
 $$
  \mu'' = \mu \cup \bigcup_{p_f \in F} \{ (p_f, x, p) \mid (p_0, x, p) \in \mu' \} \cup ( \mu' \setminus ( \{p_0\} \times \Sigma \times P ) ).
 $$
 Intuitively, from every final state of $\mathcal B$, we add the possibilty to enter the automaton
 $\mathcal B'$,
 or continue the operation of $\mathcal B$. The assumption on $p_0'$ was necessary
 to ensure that in the resulting automaton, every strongly connected component still consists of a single cycle.
 This is seen by noting that the strongly connected components are untouched by this construction. 
 It is easy to see that $U \cdot V = L(\mathcal B'')$.
 The automaton $\mathcal B''$ is nondeterministic in general, but
 by Proposition~\ref{prop:polycyclic_nea} our claim is implied. $\qed$
\end{proof} 

\subsection{Proof of Proposition~\ref{prop:polycyclic_closure} (See page~\pageref{prop:polycyclic_closure})}
\polycyclicclosure*
\begin{proof}
 If we have two polycyclic languages, by Lemma~\ref{lem:no_cycle_start} 
 we can assume they are accepted by polycyclic PDFAs $\mathcal B = (\Sigma, P, \mu, p_0, F)$
 and $\mathcal B' = (\Sigma, P', \mu', p_0', F')$ whose start states do not appear in a cycle.
 By merging those start states we get a (in general) nondeterministic automaton that accepts their
 union. More formally, construct $\mathcal B'' = (\Sigma, P \cup P' \setminus \{p_0'\}, \mu'', p_0, F'')$
 with $\mu'' = \mu \cup ( \mu' \setminus ( \{ p_0' \}\times \Sigma \times Q \} ) ) \cup \{ (p_0, x, \mu'(p_0', x)) \mid \mu'(p_0', x) \mbox{ is defined } \}$
 and $F'' = F \cup F'$ if $p_0' \notin F'$, and $F'' = F \cup ( F' \setminus \{p_0'\} ) \cup \{ p_0 \}$ if $p_0' \in F'$.
 Then $L(\mathcal B') = L(\mathcal B) \cup L(\mathcal B')$.
 By Proposition~\ref{prop:polycyclic_nea}
 this language is polycyclic. 
%
%
%
  Closure under quotients is implied, as for any accepting automaton $\mathcal A = (\Sigma, Q, \delta, s_0, F)$
  and $u \in \Sigma^*$ we have $u^{-1} L(\mathcal A) = L(\mathcal A_{\delta(s_0, u), F})$
  if $\delta(s_0, u)$ is defined, and $u^{-1} L(\mathcal A) = \emptyset$ otherwise. $\qed$
\end{proof}

\end{document}